\documentclass[11pt]{article}
\usepackage{epsfig}
\usepackage{setspace}

\usepackage{cite,url}
\usepackage{ifthen}
\usepackage{amsmath}
\usepackage{amssymb}
\usepackage{subfigure}
\usepackage{latexsym}
\usepackage{fullpage}
\usepackage{enumerate}
\newcommand{\remove}[1]{}

\newcounter{linenumber}

\newcounter{linecounter}

\newcommand{\ignore}[1]{}

\newtheorem{theorem}{Theorem}
\newtheorem{definition}{Definition}[section]
\newtheorem{lemma}[theorem]{Lemma}

\newtheorem{claim}{Claim}

\newenvironment{proof}
    { \noindent\textbf{Proof:}~}{\hfill $\Box$\\[1mm] }

\begin{document}
%\pagestyle{plain}

%\begin{frontmatter}

\title{Asynchrony from Synchrony
}
\author{Yehuda Afek\thanks{Blavatnik School of Computer Sciences, Tel-Aviv University,
Israel 69978. afek@post.tau.ac.il}
    \and
    Eli Gafni\thanks{Computer Science Department,
    University of California, Los Angeles. 3731F Boelter Hall, UCLA, LA. CA. 90095,
    USA. eli@ucla.edu}
       }

\date{\today}

\maketitle

\begin{abstract}

\ignore{%%%%%%
It is common knowledge that designing algorithms for shared-memory is easier than designing
them for message passing.  Hence the celebration over the result of more than two decades
ago that shared-memory can be implemented over message passing.
But is it true that programming with synchronous message-passing is harder than with shared-memory?

To test the validity of this thesis we first have to show that synchronous message-passing can
implement interesting models of shared-memory, e.g. showing a synchronous message-passing model
which allows to solve exactly the tasks which are read-write wait-free (RWWF) solvable.

}%%%%%%%
We consider synchronous dynamic networks which like radio networks may have asymmetric
communication links, and are affected by communication rather than processor failures.
In this paper we investigate the minimal message survivability in a per round basis that allows
for the minimal global cooperation, i.e., allows to solve any task that is wait-free read-write solvable.
The paper completely characterizes this survivability requirement.
Message survivability is formalized by considering adversaries that have a limited power to remove messages in a round.
Removal of a message on a link in one direction does not necessarily imply the removal of the message on that link in the other direction.
Surprisingly there exist a single strongest adversary which solves any wait-free read/write task.
Any different adversary that solves any wait-free read/write task is weaker, and any stronger adversary will not
solve any wait-free read/write task. 
ABD \cite{ABD} who considered processor failure, arrived at an adversary that is $n/2$ resilient, consequently can solve tasks, such as
$n/2$-set-consensus, which are not read/write wait-free solvable.
With message adversaries, we arrive at an adversary which has exactly the read-write wait-free power.
Furthermore, this adversary allows for a considerably simpler (simplest that we know of) proof that the protocol complex of
any read/write wait-free task is a subdivided simplex, finally making this proof accessible for students with no
algebraic-topology prerequisites, and alternatively dispensing with the
assumption that the Immediate Snapshot complex is a subdivided simplex.

\end{abstract}

\noindent\textbf{Keywords:} shared memory, distributed algorithms, wait-free, subdivided
simplex, asynchronous computability.
%two front agreement

%\noindent\textbf{Regular presentation}
%\renewcommand{\baselinestretch}{1.7}\normalsize

\newpage

\section{Introduction}
In \cite{ABD} Attiya Bar-Noy and Dolev showed that message passing
can simulate shared memory, by implementing read and write
operations when a majority of the nodes do not fail. But this
message passing model, in which majority of processors are alive, is
stronger than wait-free read write, as it can solve
$n/2$-set-consensus in it.
\ignore{Also, the simulation in \cite{ABD} is a wait-free simulation
\ignore{(modulo the assumed number of live processors) }
of the read and write operation. To solve tasks we need only non-blocking rather than
wait-free simulation.}
Here we address and resolve the following question: is there a message passing model
that is exactly equivalent to wait-free read write solvability?  That is, a network model
that can solve any task solvable wait-free in read write shared memory, and nothing more.

It was assumed that message-passing is not interesting when majority of processors can fail-stop, since then
network partition may lead to inconsistency.  This lead to investigating processor failure as source
of non-determinism.  In various models of iterated shared-memory \cite{BG97}, like old soldiers, processors do not fail,
they just fade away by being consistently late from some point on.  Iterated models have been found to be useful since they
can be thought of as sequence of tasks, but since to date they were proposed in the context of shared-memory
the logic synchrony embedded in them was not apparent.

Thus the impetus behind this work is to expose upfront the ``logical synchrony'' of the iterated shared-memory model
by real synchronous message passing models in which not receiving a message by the end of a round means that a message is not
forthcoming, and, on the other hand, show that all models of shared-memory can be investigated
within this synchronous message-passing framework.  We provide here such a synchronous model
that captures exactly what is wait-free computable in an atomic read and write memory.
\ignore{Thus we address and resolve the following open question: is there a message passing model
that is exactly equivalent to wait-free read write solvability?  That is, a network model
by which any wait-free read write task can be solved, and nothing more.}%%%%%%%
In our models as well as in the iterated models processors do not fail, only some may be invisible to the rest
from some point on, which makes them ``faulty'' with respect to the rest.

\ignore{%%%%%%%%%%%%
Yet asymmetric synchrony has been studied ``in disguise.''
The models of iterated computations \cite{BG97} is a computation that is ``logically synchronous'' rather than expressively so.
In hindsight, the easiest way to understand what is the model of say iterated snapshots \cite{SergioOpodis},
is to think of synchronous complete network with a TP-complete adversary whose RCG has the constraint
that for each $p_i, p_j$ the incoming edges are related by containment.

Thus, we enrich the study of iterated computations by considering more refined communication patterns
than the ones implied by shared-memory. In fact, if we consider an iteration to be a task, then the model
of iterated snapshots is a concatenation of tasks in which processors submits abstract item and get a set
of submitted items. The snapshot tasks establish that the sets must be related by containment. More generally,
until now these tasks were considered at the granularity of shared-memory. But shared memory establishes
all sorts of constrained relations among the returned set. Here, we forgo the shared-memory and ask as
to why not study these tasks as tasks abstract tasks rather than only ones that are the result of some
shared-memory implementation.
}%%%%%%%%%%

We consider a synchronous round based dynamic complete network in
which processors do not fail but individual messages do. Like in a radio network,
messages may be asymmetrically dropped.
We consider the power of such networks to compute when in each synchronous round all
processors send messages to all their neighbors, and an adversary
can purge some of the messages.   Removing a message on a link by the adversary in one
direction does not imply the removal of the message in the other
direction.
The adversary is a predicate: What combinations of messages may be dropped in a round.
The predicate does not spread across rounds.
The adversary ``power" in each round is the same and is independent of what it did in the previous round.
\ignore{What happens in a
round in message loss is independent of what happened in the
previous round.}
An adversary $adv_w$ is weaker than adversary $adv_s$ if the predicate $adv_w \subset adv_s$, i.e.
every message failure combination in $adv_w$ is a message failure combination allowed to  $adv_s$,
but not vice versa.

We show an adversary called $TP$\ignore{W} (Traversal Path)
such that the resulting model is equivalent to read write wait-free (RWWF).
Moreover, adversary $TP$ is the most powerful among all adversaries that
can solve any RWWF solvable task.  

\ignore{What is the benefit of investigating shared-memory through synchronous message passing?
}
Following the presentation of the $TP$ adversary we exhibit how the derivation of the necessary part of
the celebrated Herlihy-Shavit conditions \cite{hs99},
for wait-free solvability, namely, that for a task to be RWWF solvable it has to color a subdivided simplex, is
``for-free.''  Originally, this result called upon heavy machinary in Algebraic-Topology.
Subsequently, it was simplified in \cite{BG97} with the
use of iterated Immediate-Snapshots, which still requires the
assumption that the Immediate Snapshot complex is a subdivided simplex (not obvious in high dimensions).
Here the proof is quite elementary.

This line of investigation also leads to a new classification, left to come, of network topologies.
Given a network, and an adversary that allows for solving RWWF solvable task, what is the power
of this adversary to solve tasks which are not RWWF solvable.
Compared to complete network, a lack of a link in a network can be viewed as a constraint on the 
adversary to always fail the two messages of a link, reducing its predicate and consequently
making it potentially strictly more powerful to solve tasks. 
To arrive at an adversary that solves exactly RWWF we use
here the topology of complete network.
Yet if we take a network which is a single simple undirected path, an adversary that
solves RWWF necessarily soles $2$-set consensus.  A processor can output
one of the two end nodes of the line.  

%%%%%%%%%%%%%%%%%

\subsection{related work}

The closest we can recall studies of computational power against synchronous message adversary
is in the context of Byzantine agreement in the presence of an adversary that corrupts messages
rather than processors \cite{DR85}. Similarly, many papers touch on the communication requirements to
achieve consensus.  Indeed, when a communication link is either up in both direction or
down in both directions, we either have connectivity and consequently consensus, or we have disconnected
components with no coordination among them.

\ignore{Thus, the novelty of this paper is in considering asymmetric adversaries. Ones that can delete messages in one
direction but not necessarily in the other direction.

Thus again in hindsight the celebrated result of Attiya, Bar-noy, and Dolev \cite{}, is to ask the following question:
Give an adversary that will guarantee that at at least majority of processors are alive, i.e. an adversary
that can solve exactly ``minority''-resiliency. They propose a TP-complete adversary in which all
nodes have a fan in of at least $n/2$ and at least $n/2$ processors have a fan-out of at least $n/2$.
As said, in the context of processor failure solving RWWF implies solving minority resilient tasks.
Indeed same holds for the minimum failure detector $\Sigma$ for read-write introduced in
\cite{Rachid,Delapote,Hugues, JACM}. It not only allows for solving RWWF but also any minority-resilient task.
}

It is easy to see that read write shared memory implements the message passing model,
each processor sends a message to its
neighbor by writing the message into a dedicated single writer single reader register.
In the other direction Attiya Bar-Noy and Dolev \cite{ABD} show how to implement a single writer multi-reader register
in an asynchronous network in which a majority of the processors do not fail, essentially by
giving a sequence number to each value written and ensuring that each
value written or read is documented in at least a majority of the processors (like a quorum system).
However, this network model in which a majority of the processors do not fail is computationally stronger than
read write wait free, it can solve $n/2$-set consensus.  After hearing from $n/2$ processors a processor outputs the
minimum value it has heard about.

In \cite{klo10} Kuhn, Lynch, and Osman study dynamic networks that are also governed by an adversary.
In their case the adversary can erase edges (communication in both directions or none) and they restrict
the adversary to various types of eventually connected network \cite{ag88,ae84}, called T-interval connectivity.
While such a network can solve the consensus problem they investigated the complexity of key distributed computing
problems such as determining the network size, and computing any function on the ids of the processors (e.g., leader election =
consensus).

That the protocol complex of models that solve exactly RWWF task
contains a subdivided simplex was first established by Herlihy and
Shavit \cite{hs99}\ignore{ using the Meyer-Viatoris Lemma from
algebraic-topology \cite{} among other topological tools}.
In \cite{BG97}  this result was shown
using the iterated model notion and immediate snapshot tasks. It
assumed without proof that the immediate snapshot task is a
subdivided simplex.\ignore{ This fact was proved only recently \cite{Linial
- personal communication}. Here, we derive this result in the most
simple way. Showing that inductively a read-write algorithm takes a
set of 1-dimensional faces in a subdivided simplex, that do not
share a simplex. It subdivide this 1-dimensional faces and cone the
edge thus divided into 3 segments with the rest of the points of the
simplexes in which the edge resides.}

\ignore{%%%%%%%%%%   DISC Material
Finally, we use the same technique to show that as long as the
adversary is not forced to make the whole RCG a single SSC, then the
protocol complex of the resulting model is connected. Thus we get
also FLP \cite{} for free. Currently the simplest way of deriving
the FLP result is via the BG simulation \cite{}. By itself important
but not elementarily derived result.
}%%%%%%%%%%%%%%%

\paragraph{Paper organization:} The paper is organized as follows.  In the next section we discuss
the model of synchronous dynamic network with message adversary in
more detail.  We then as a worm-up show ``procedurally'' how
our TP adversary implements Read-Write.  In a subsequent section we rely
on previous work that shows that iterated snapshots (IS) model
\cite{BG97,SergioOpodis} solves exactly any task that is RWWF solvable.
Then in a
declarative manner of tasks implementing tasks, we show TP in
complete networks is equivalent to IS.  Subsequently, in Section
\ref{section:subdivided}, we introduce an artificial model which can be easily seen
to be equivalent to TP.  We then show in elementary inductive way how
the artificial model gives rise to a protocol complex \cite{hs99} which
is precisely a subdivided simplex.  In Section \ref{section:subdivided} we
detail the construction for $n=3$ and in the appendix the construction is detailed
for an arbitrary $n$.  Finally we close with conclusions.

\section{Model}
\label{sec:model}

This paper deals with the relations between two basic models in distributed computing, the read write shared memory, and
the synchronous message passing with message delivery failures.  In either model there are $n$ processors, $p_1, \ldots p_n$.
In the shared memory we consider the standard read write shared memory model \cite{H91} in which all communication
between processors is via writing to and reading from shared single-write multi-reader atomic registers.

The network model we assume is
synchronous complete network (unless stated otherwise) in which in each round each processor sends its
entire history to all other nodes in the network.  In each round an adversary may purge a subset of the
messages sent.  All other messages are received by their destination by the end of the round.
Depict a message successfully delivered from $p_i$ to $p_j$ in a round as a directed edge from $p_i$ to $p_j$.
The collection of directed edges in one round is called the round communication graph (RCG).
The adversary is specified by a property that must be satisfied by any directed graph RCG it can create.
In the strongly connected (SC) adversary, in each
round the RCG must be a strongly connected graph spanning all the nodes.
In a Traversal Path (TP) adversary, the RCG in each round must contain a directed (not necessarily simple) path
starting at one node and passing through all the nodes.
A dynamic network ruled by an adversary $\Gamma$ is called $\Gamma$-dynamic network.

We extensively us an adversary we call TP-complete. This adversary is defined only with respect to underlying
complete network. Its predicate is that RCG contains a tournament. In other words, of the two messages on a link
sent in a round TP-complete can purge at most one. For complete networks TP-complete and and TP-dynamic are
shown equivalent. TP-complete captures one property of shared-memory, namely, two processors cannot miss each
other. The model of SWMR collect has more properties than TP-complete (these properties can be called ``fat immediate snapshots'').

We consider only ``anonymous'' adversaries. Whether a graph is a valid RCG for the given adversary is invariant to renumbering of nodes.

A task \cite{hs99} is a distributed computational problem involving
$n$ processors.
Each participating process starts with a private input value,
exchanges information with other participating processes and eventually outputs a value.  The task is specified by a relation
$\Delta$ that associates with every input vector (one element per
participating processor) a set of output vectors that are allowed
given this input. See \cite{hs99} for a more formal definition of a task.

We now describe the meaning of computation in synchronous dynamic
network. An input is an abstract ``item.'' Computation evolves in
synchronous rounds. In the first round each processor sends a
message consisting of a pair $(my id, my item)$ to all neighbors.
Some messages get through some are deleted by the adversary. The
messages that make it through have to comply with the predicate
$\Gamma$ that defines the RCGs that the adversary must maintain.  In
every subsequent round inductively a nodes sends a message with all
its history to all its neighbors. A protocol to solve a task $T$
with an adversary $\Gamma$ is associated with a number $k$, and
after $k$ rounds each processor takes its history and maps it to an
output of $T$. The protocol solves $T$ if under the condition that
in every round the communication graph RCG agrees with $\Gamma$ then the
outputs of all processors are valid combination through $\delta$
with respect to the initial items which are the inputs.

Interchangeably, we will view a round as a task: The abstract ``item'' is the processor id,
and the output of a node/processor is a set of ids. The task is then defined through a predicate
on the combination of returned sets. Obviously, any variant of a model of ``iterated shared memory''
can be captured by such task, and consequently as an adversary in our network.
Notice that the view of task as being ``invoked'' by a processor or not is mute when we compute in synchronous
networks. Processors do not ``crash,'' they always are together in a round. It is just that some may not be seen
in a round. Thus, say, $t$-resiliency is an adversary whose RCG has a SCC of size at least $n-t$, and this SCC is
a source SCC in the SC graph of the RCG.

\ignore{%%%%%%%%%%
\subsection{Synchronous Dynamic Network Results}

\begin{enumerate}

\item
The Strongly Connected Component graph (SCCG) of a TP-adversary is
an acyclic tournament. (We call the single source node in SCCG-TP to
be the ``source SCC'' and the single sink node ``sink SCC'')

\item
DISC item !!! A TP-adversary with the additional constraint that the source SCC is
of size $n-t$ solves all $t$-resilient tasks. In case the network is
complete it solves only these tasks.

\item
To solve consensus the source, and the single SCC, has to be of size
$n$, the number of nodes. Consequently, we automatically obtain FLP
\cite{}, as the adversary for 1-resiliency is not required to have a
single SCC.

\item
In a complete network the protocol complex of a model of computation
we call TP-complete-pairs which is shown equivalent to TP-adversary
and therefore to RWWF, is a subdivided simplex.

\end{enumerate}
}%%%%%%%%%%%%%%%%%

\section{TP-dynamic network implements wait-free read write}
\label{sec:king}
In this section as a warm-up we show a direct implementation of asynchronous read-write wait-free by a TP-complete dynamic
network. To solve a task RWWF it is enough that the emulation is non-blocking.
As some processor takes enough steps it eventually gets an output.
From there on, processors can ``ignore'' it, and some other processors will progress, continuing inductively until all
processors terminate.

In the next section we show that TP-dynamic network is equivalent to TP-complete dynamic network, each can implement the other.
Thus here we implement the RWWF in a TC-complete dynamic network.
Thus TP-dynamic network is potentially stronger than the RWWF model.  In the next section (Section \ref{sec:tasks})
we also show that TP-dynamic is equivalent to
iterated snapshot model which is equivalent to RWWF solvability \cite{SergioOpodis}.  Thus deriving both directions.
We nevertheless go through the direct implementation as a warm-up and making the paper self contained.

To simulate the execution of a RWWF run of task $T$ in a TP-complete dynamic network
we will simulate the simplest read write model, a run in which all writes and reads are
to single writer single reader registers.  As we know such a run implements any RWWF execution.
Each processor $p_i$ starts with its input in $T$ and registers its first write, $<p_i,w^i_1,1>$.
Its write, $w^i_1$ is size $n$ vector with the values it writes to each of the single writer single reader
registers it writes to, with a $\bot$ for registers in which it does not write at this point.
To simplify notations and exposition we write just $w^i_x$ representing this vector, i.e., each $w$ below is
a vector with the write values $p_i$ writes to each other SWSR register.

Each processor maintains a vector of triplets $<p_j, w^j_k, k>$ indicating that the most advanced
write of $p_j$ it has heard is $w^j_k$, the $k$th write of $p_j$.
At the end of each round a processor updates its own vector with most advanced
write values it finds for each other processor in all the vectors it received.
In the messages sent, each processor sends the vectors it received (directly or indirectly) for other processors,
and its own vector after it has been updated.  For each other processor, $p_k$ it sends the most advanced
vector it has received for $p_k$ (there is a simple total order on the vectors sent by each processor
in the course of execution).
Thus each processor keeps a copy of the most updated vector it knows for each other processor.
Therefore, a processor can tell which writes, to the best of its knowledge, each other processor has already received.

As we prove below after a finite number of rounds
there must exists at least one node $s_1$, such that its write value, $<s_1,w^{s_1}_1,1>$
has reached all other nodes in the network,
and node $s_1$ knows that its first write value has reached all others.

At this point, when a node knows that its write value has reached all other nodes it has finished its
first write operation
and performs a read operation.  For each $p_j$ the value it reads is $w^j_k$ where $k$ is the highest
index $l$ among all the triplets $<p_j,w^j_l,l>$ in its set.  If no such pair exists it returns
the value $\bot$.  The read by $s_1$ is linearized at that point (nobody could have written at that
point a higher index write since to terminate a write the writer must ``know''
everybody has its last write in their set and $s_1$ does not have it prior to this point).
We linearize a write operation at the time its value was first read by some processor.
Since we write to SWSR registers this linearization is legal.

Continuing inductively, following a round in which $p_i$ finished its $k$th
write operation since at that round it also read, it calculates it $w_{k+1}$ write values
(one for each SWSR register) and inserts the pair $<p_i,w^i_{k+1},k+1>$ into its vector.
Thus inductively, for any finite $r$, after a finite number of iterations there must be at least one node $s'_1$
such that $s'_1$ has completed $r$ write and read operations.
Since task $T$ is a wait-free task, there is a finite $r$ such that after $r$
iterations $s'_1$ outputs.  From now on it cannot ``calculate'' its next write because it has finished its
execution in task $T$.
Indeed, if its vector when it outputs is $my-out-vector$ it inserts a pair $<s'_1,my-out-vector,.>$,
and from now on becomes just
a repeater participant. Never finishing a new write.
Processors that see the pair $<s'_1,my-out-vector,.>$ can calculate the output of $s'_1$, and from now on
in their ``know'' function can exclude $s'_1$.

A key observation in our simulation is that
from this point on the other nodes do not expect to hear from $s'_1$ and node $s'_1$ becomes a relay,
it forwards in each round to all its neighbors the union of all the messages it has received in the previous round.
Therefore, once such a subset $S_1$ of nodes has completed their computation in $T$,
there must exists at least one node $s_2 \in V\setminus S_1$ such that
its value reaches all the nodes in $V\setminus S_1$, and $s_2$ knows that.
Thus inductively a subset
$S_2$ completes $r$ iterations and finishes its computation of $T$ and outputs.  This process continues until all the nodes complete the
computation of $T$ and output.

To prove the above we need to show that after $\ell$ rounds for some finite $\ell$, there is a
node $p_i$ in a dynamic network with adversary TP-complete such that $p_i$'s write value has reached
all other nodes and node $p_i$ knows that, we call node $p_i$ the {\em king} node.
Each node uses the following condition to
decide that it is a king node:

\begin{theorem}
\label{thm:king}
\ignore{%%%%%%
Node $p_i$ is king if for $n$ rounds node $p_i$ does not hear from any processor that did not report to have heard $p_i$'s write.
in other words:

Node $p_i$ is king if for $2n-1$ rounds all the vectors in messages it receives report to have heard
$p_i$'s last write in them (i.e., for $2n-1$ rounds $p_i$ does
not hear from processors that do not report to have heard from $p_i$).}

Node $p_i$ is king in round $t$ if in round $t$ $p_i$ does not hear directly from any other processor that
does report to have heard $p_i$'s write value.
I.e., from each other node $p_j$ in round $t$ either $p_i$ does not get a message from $p_j$
or $p_i$ gets a message from $p_j$ in which $p_j$'s vector contains the last write of $p_i$.

\end{theorem}

\begin{proof}
By \cite{lan53} there is a king-L (kink-L is a node that all other nodes are reachable from it in
a length at most $2$ directed path) in a static tournament directed graph.  Therefore eventually after enough rounds of the dynamic
TP-complete network there is a node that is a king as defined in the theorem.  Thus eventually such a node
exists.  Clearly the condition of the theorem tells the node it is a king because the nodes it did not hear from
in$t$ certainly have heard from it in $t$.  The other nodes directly report to the king to have received its write value.
\ignore{%%%%%%%%%%%%%%%
By contradiction.  Consider node $p_i$ at a round in which the theorem condition is satisfied, and
let $p_j$ be a processor that did not receive $p_i$'s write value by that round.
Clearly $p_i$ does not hear from $p_j$ (directly or indirectly) in the last $n$ rounds.  Let $R_j$ be the set
of processors that $p_j$ heard from them, and $\bar{H_j}$ the set of processors that did not hear from $p_j$ in the last $2n-1$ rounds.
Since by PT in each round the RCG contains a traversal path then in each of the last $2n-1$ rounds either $R_j$ or $\bar{H_j}$ increases by
one, i.e., either $p_j$ received $p_i$'s write value, or $p_i$ receives $p_j$'s vector, a contradiction.
}%%%%%%%%%
\end{proof}

\ignore{%%%%%%%%%%%%
Similarly we can prove that when discarding a node, such as a node that has produced an output, node $p_g$ is king if for $n$ rounds,
after excluding any node that has output (has completed its computation):
\begin{enumerate}
\item
Node $p_g$ does not hear about any new other node, and
\item
Any node $p_j$ that node $p_g$ has heard about also reported to $p_g$
that it has heard from $p_g$ or $p_j$ is reported to have produced an output.
\end{enumerate}
I.e., in both point $1$ and $2$ above nodes that have reported to terminate the computation are not taken into account.

Thus to simulate the RWWF execution of task $T$ each node $p$ repeatedly, in each round, receives messages, adds their content to $p$'s history and
sends its entire history to its neighbors.  If $p$'s history does not contain information about any new node for $n$ consecutive rounds, and
any node $p$ has heard about also reports in $p$'s history to have heard the last value $p$ is writing in the execution of $T$, then node $p$ has completed
another write operation.  It then performs a read based on all the messages it has received.  Following that read $p$ either outputs and reports its state
to be ``{\em done}", or $p$ generates a new write value in $T$ and repeats the above procedure.
It follows from Theorem \ref{thm:king} that eventually all the nodes terminate and produce their output in $T$.
}%%%%%%%%%%%%%%

\paragraph{TP is the strongest adversary in which RWWF can still be simulated.}  There are many different ways by which
the TP adversary can be made stronger.  For example, allow it to remove one more message, or in any RCG there is a simple
{\em undirected} path spanning all the nodes but not a directed path.  However notice that in any way by which TP can be made stronger there is
at least a pair of nodes, $p_i$ $p_j$ such that neither can guarantee to communicate with the other.
Thus there is no way to simulate either that $p_i$ wrote and $p_j$ read what $p_i$ wrote or vice versa.
Therefore, any read-write wait-free task in which either $p_i$ returns $p_j$ or vice versa cannot be implemented, e.g., snapshot.

\section{TP-dynamic network equivalent to wait-free single-writer-multi-reader shared memory}

\label{sec:tasks}

In \cite{BG97} and \cite{SergioOpodis} it has been shown that
the iterated snapshot model, called IIS, solves only RWWF tasks.
Thus IIS is equivalent to SWMR RWWF model.
In the snapshot task \cite{ATOMIC}
each processor $p_i$ submits a unique ``item'' and returns a
set  $S_i$ of submitted items which includes its own, and the returned sets are
related by containment.  In an iterated model of task $T$ in which items are
submitted and set of submitted items are returned, called IT, there is a sequence
of independent copies of the task, $T_1, \ldots, T_j, \ldots $ and each process goes through
the sequence in order.  It starts
by submitting its input to $T_1$ and then inductively submitting to $T_j$ with its output from $T_{j-1}$.
Different iterated models have been considered, where $T=IS$, the iterated snapshot $AS$ model $T=AS$,
as well as iteration of the collect task \cite{collect1,AAFST}.

Here we show that TP-dynamic network is equivalent to the iterated snapshot (IAS)
model which together with the result of \cite{SergioOpodis} proves that TP-dynamic
is equivalent to RWWF.
\ignore{as we find fixed number of rounds of TP-dynamic that
implements the task of one-shot snapshot.}

\paragraph{The snapshot task implements a round of TP-Dynamic network:}
The items submitted to the snapshot are the messages in a round.
If $p_i$ in the snapshot task returns the item of $p_j$ (or $p_j$
for short) we consider it that the message from $p_j$ to $p_i$ was
successful.  Since snapshot is an instance of shared memory
collect processor do not miss each other, and RCG contains a tournament.
Thus, the snapshot task implements one round of TP-complete.
Since each tournament contains a directed path spanning all the nodes \cite{BNaor},
we get that a tournament is an instance of TP-dynamic.
Thus, all tasks solvable by TP-dynamic are solvable in IAS.

\paragraph{TP-Dynamic network implements a snapshot:}
To show that TP-dynamic implements snapshot we let the model
TP-complete mediate between them.
\ignore{We first show that TP-dynamic implements TP-complete then that TP-complete implements
snapshot.}

We first show that TP-dynamic implements TP-complete in $2n-1$ rounds.
In each round a processor $p_i$ sends to all its neighbors the set $H_i$ of inductively all the ids
it has heard so far, starting by setting $S_i=\{p_i\}$ in the beginning of first round, and sending $S_i$.
At the end of each subsequent round it just sets $S_i$ to the union of his set with all the sets it received in the round.
\begin{claim}
After $2n-1$ rounds for every $p_i$ and $p_j$,
either $p_i \in S_j$, or $p_j \in S_i$, or both.
\end{claim}

\begin{proof}
Let $H_i$ at the end of a round be the set of nodes $p_k$ such that $p_i \in S_k$.
If neither $p_i \in H_j$ nor $p_j \in H_i$, then, using Sperner Lemma for dimension 1, in the next round
of TP-dynamic a message is successful from either a node in $H_i$ to a node not in $H_i$ or
from a node in $H_j$ to a node not in $H_j$, or both.  Thus in each round the size of at least one of the sets
$H_i$ and $H_j$ increases by one.
\end{proof}

To show that TP-complete implements snapshot we go through $n$
rounds of TP-complete.
As before, in each round a processor $p_i$ sends to
all other nodes all the ids, $S_i$, it has obtained by now.  Let
$S_j[k]$ be the the set $S_j$ at the end of round $k$.
Then, if at the end of round $\ell$, $|S_j[\ell]|=\ell$, $p_j$ returns
$S_j[\ell]$.

\begin{claim}
The sets $S_i, i=1,\ldots n$ thus returned in the above algorithm are snapshots.
I.e., $\forall i, p_i\in S_i$ and $\forall (i$ and $j), S_i \subseteq S_j$ or $S_j \subseteq S_i$.
\end{claim}

\begin{proof}
By induction.
Assume that by the end of round $k-1$ at most $k-1$ processors returned
and they returned snapshots, and all the processors which did not
return have sets containing the largest (snapshot) set returned so far and all
the sets are at least of size $k$ at the beginning of round $k$.

Observe that the inductive assumption holds at the end of round $k=1$.
First we show that at round $k>1$ only a single set can be returned.
Since each of the pair $p_i, p_j$ has a set of size $k$ or more (i.e., $|S_i[k-1]|\geq k$ and $|S_j[k-1]|\geq k$),
and one of
them at least hears from the other, then if processor $p_i$ returns
then it saw only sets identical to his own ($S_i[k-1]$), so either $p_j$ does not return if
it has a different set, or it returns since it has the same set.

If $p_i$ returns, then it received messages only from processors with
a set identical to his ($S_i[k-1]$), thus $|S_i[k-1]|=k$) that
$p_i$ returns was sent to all processors who do not have an identical
set, consequently they either heard already of $k+1$ items
or they heard about $k$ items and $S_i$ adds at least one more since it
is different. Thus the hypothesis that the set that
continue to round $k+1$ contain $S_i$ if $S_i$ was returned is maintained.
establishing containment.

Since processors heard about themselves and the maximum size
returned by now is $k$ then the number of processors returning is at
most $k$.
\end{proof}

Thus we have established that TP-dynamic is equivalent to RWWF.

\section{TP-dynamic network colors a subdivided simplex}

\label{section:subdivided}
To show that the outputs of a multi round execution of TP-dynamic network colors a subdivided simplex we show
that another ``adversary'' (now we spread the predicate over rounds) called TP-pairs implements TP-dynamic, and show that the outputs of TP-pairs color a subdivided
complex.
We show the equivalence of TP-dynamic to TP-pairs,
by showing that TP-pairs is equivalent to TP-complete.
That the latter is equivalent to TP-dynamic has been shown in  \ref{sec:tasks}.

Recall the definition of TP-complete:  In the TP-complete adversary in each round on each
\ignore{ $p_i$-$p_j$} edge at least one message is delivered in one direction or both directions.
In the TP-pairs adversary we spread TP-complete over $n(n-1)/2$ rounds where in each such round
a message is sent in both directions of one unique edge and at most one of these two messages may be purged.

Clearly, TP-complete implements TP-pairs by going $n(n-1)/2$ rounds and at each round
ignoring anything which is not associated with the particular edge of the round. To see that TP-pairs implement TP-complete,
we go for $n(n-1)/2$ rounds where in each round in which a processor send it sends
what it sent the first time it was scheduled. At the end of the round a processor just collects
all the messages it received.

\subsection{The TP-pairs Protocol-Complex at an End of a Round}

Here we show that the protocol-complex of TP-pairs is a subdivided complex.
Keeping the exposition simple, we avoid unnecessary formalism and notation by restricting
this section to $n=3$.  The $n=3$ case generalizes to higher $n$'s in a straightforward way.
For completeness we repeat the construction given here, but for arbitrary $n$ in Appendix
\ref{app:subdivided}.

Consider all the possible local states of the $3$ processors $p_0, p_1, p_2$ after round
$r$ and make a graph $G_r$ out of it.  The nodes of $G_r$ are the pairs of processor-id
and its possible local state at the end of round $r$.  Two nodes are connected by an edge
between them if there exists an execution $E$, that is an instantiation of the TP-pairs adversary
in rounds $1$ to $r$, such that the corresponding processors are in the corresponding
states.  Assume inductively that $G_r$ is a 3-colored triangulated triangle.

The process starts with a triangle of the 3 processors in their initial state.

Let round $r+1$ be a round in which messages are sent (only) on the edges of type $(p_i, p_j)$.
How does the different behaviors of the
adversary in round $r+1$ change the graph?  In place of each
$(p_i, p_j)$ edge we get now a path of $3$ edges $(p_i, p_j'),(p_j',p_i'),(p_i',p_j)$. The
first node in the new path is a $p_i$ node with a state in which $p_i$ appends to its
local state that it did not receive a message from $p_j$ in round $r+1$.  Analogously, the
node $p_j$ at the other end of the path.  The nodes in the middle, $p_j'$ and $p_i'$, are
nodes in which the corresponding processor appends to its state the content of the
message it received in round $r+1$.  The nodes corresponding to the third processor $p_k$
has only one new incarnation $G_r$ to $G_{r+1}$ as $p_k$ records that did not receive any message
in the round.
Obviously a node of type $p_k$ that was connected to
a $(p_i, p_j)$ edge at the end of round $r$ is now connected to the four nodes of the
path that replaces the $(p_i, p_j)$ edge in $G_r$.  See Figure \ref{fig:2ad-1}(a).

\begin{figure}[htbp]
\centering \epsfig{file=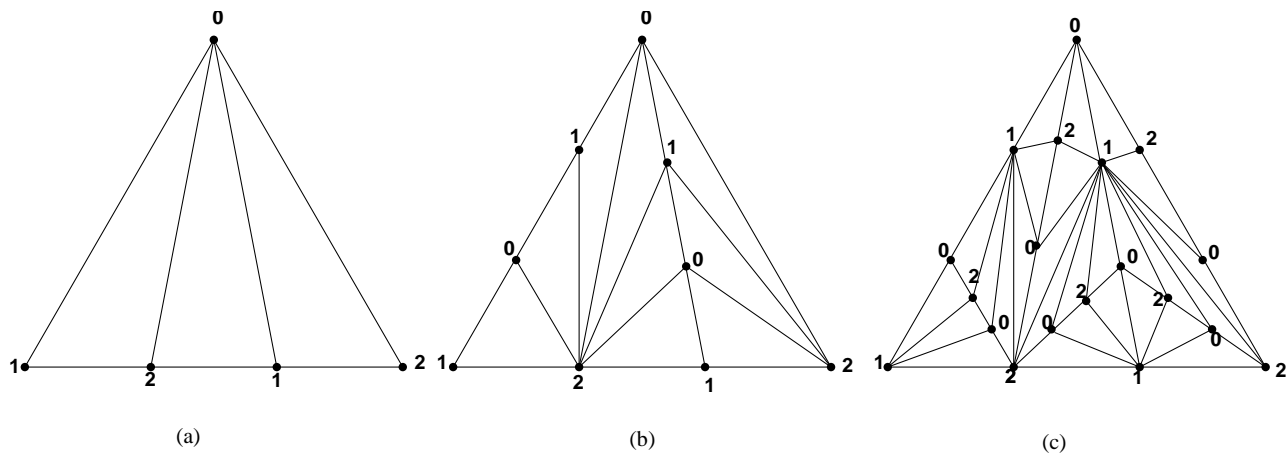, width=17cm} \caption{A $3$ processors, $3$ rounds
subdivision example.  The order edges on which messages where exchanged is $(1, 2), (0, 1), (0, 2)$.  (a), (b) and (c)
are the corresponding first, second and third split operations.} \label{fig:2ad-1}
\end{figure}

For example, consider this process of graph evolution by drawing $G_{r+1}$ in the plane
(Figure \ref{fig:2ad-1}(a)). Initially the graph is the triangle of the initial states.
Assume the first round of the TP-pairs schedule is sending messages on edge $(1, 2)$.
To construct the
graph corresponding to this round take the $(p_1,p2)$ edge and replace it by a path by
planting two middle nodes on the edge, denoting the new local states by the original
$(p_1, p_2)$ to get an alternating path of $p_1, p_2, p_1, p_2$.  We now connect the
middle nodes with node $p_0$, and we got $G_1$.

Inductively, we embed the initial triangle in the plane, it is 3-colored by the 3 ids.
In any inductive step we embed node on an edge and connect them to the third 
node in the triangles the edge is in. Obviously we have an embedding and the corresponding triangilated
triangle is 3-colored, and the original edges of the triangle we started with are now a face which is 
2-colored by the colors $p_i,p_j$ defining the original egde.

\section{Conclusions}

We have taken a step beyond iterated models to consider pure
synchronous message passing models with message adversary. 
Just this simple realization is interesting as it shades a big light
on the inherent logical synchrony of the iterated models. We have then
done the first step in showing the benefit of the new paradigm, by
showing there is an adversary that in a complete network is exactly
RWWF.  We have restricted our notion of models to tasks
invoked by a processor with its id (our abstract ``item'' in the
body of the paper) as an input and returning a subset of the ids of
processors that invoked the task. It is known that all
``reasonable'' set of runs of shared memory can be captured by such
tasks. Obviously any such task that follows from shared memory
(e.g., write your id, scan until at least observed $n-t$ ids, return
the scan = $t$-resiliency) can be captured  by an adversary.  But is
the reverse true? Is every adversary possess exactly the
computational power of some ``procedural'' shared-memory task? We
suspect the answer is yes. Finally, we have introduced a new question (which is
obviously decidable) is to came up with an algorithm that given an
adversary will output its computational power.

In a complete network it is known that Immediate Snapshots is the weakest adversary that is just RWWF,
and TP-dynamic was proved the strongest. Thus the adversaries of RWWF are sandwiched above and below.
It is true for all adversaries, i.e. do all equivalent adversaries constitute a lattice?

Last but not least. Elsewhere, in a companion submission, we have realized 
that our understanding of the notion of solving a task in shared-memory read-write
model, is not crystal clear. Solving tasks in models the processes do not ``fail-stop''
but are rather just late, puts theoretical distributed computing on much firmer grounds.
We have added richness to this domain, by associating it with more ``practical''
models, then just artificial iterations over a bank of shared-meories.

%%%%%%%%%%%%%%%%%%%%%%%%%%%%%%%%%%%%%%%%%%%%%%%%%%%%%%%%%%%%%%%%%%%%%%%%%%%%%

\appendix

\section{The TP-pairs subdivided complex} \label{app:subdivided}

For the sake of completeness we follow here the arguments given in Section
\ref{section:subdivided} and show somewhat more formally that the TP-pairs model when executed
for $k=f(j)$ rounds implies a subdivision of the input complex.  The argument here is for
arbitrary $n$ and is similar to that in Section \ref{section:subdivided}.  Each elementary simplex
in the subdivision corresponds to a particular $k$ rounds execution of TP-pairs, i.e., a
particular instantiation of the adversary.  The number of iterations, $k=f(j)$ is taken
as the number of rounds required by the TP-pairs in the worst case to emulate a complete
wait-free execution of a task under consideration (see Section \ref{sec:king}).

W.l.o.g, instead of subdividing the input complex we create a subdivision ${\cal T}$ of
an $n$-vertex simplex ${\cal P}$, corresponding to an input less initial configuration
simplex.   Every vertex $v$ in $\sigma$ is associated with a processor's name $\chi(v)$.
Distinct vertices at the same simplex of ${\cal T}$ are associated with distinct
processors, so that $\chi$ is a proper coloring of the graph which is ${\cal T}$'s
one-dimensional skeleton.  Clearly, initially $\chi$ is a proper vertex coloring of $\cal
P$. In addition, the construction of ${\cal T}$ is such that $\chi$ has the {\em Sperner
property} i.e., if the vertex $y \in V({\cal T})$ is on ${\cal P}$'s boundary and if
$\sigma$ is the lowest-dimensional face of ${\cal P}$ that contains $y$, then $\chi(y) =
\chi(x)$ for some vertex $x$ of $\sigma$.

\begin{lemma}
\label{lem:tri} Given the TP-pairs model on $n$ processors, and $k$ the number of rounds
executed by the model, then there is a subdivision ${\cal T}$ of the $n$-vertex simplex
${\cal P}$ with proper coloring $\chi$ of $V({\cal T})$, the vertices of ${\cal T}$. Such
that each simplex in the subdivision represents the final states of the $n$ processors at
the end of the execution.  The coloring $\chi: V({\cal T}) \rightarrow \{0,\ldots,n-1\}$
satisfies Sperner's condition.

\end{lemma}

\begin{proof}  We go through a sequence of subdivision refinements, ${\cal T}_i$, $i = 0,\dots, k$,
where, ${\cal P} = {\cal T}_0$,  ${\cal T}_{i+1}$ is a subdivision refinement of ${\cal
T}_i$ that corresponds to all possible behaviors in the $i'th$ round of the TP-pairs
adversary, and ${\cal T}_k = {\cal T}$.  Thus we go repeatedly over the sequence of
rendezvouses in the order in which the TP-pairs scheduler goes, for $k$ steps. In each
step we consider the three possible behaviors of the adversary in the corresponding
rendezvous.   The resulting ${\cal T}$ is finite, with $3^k$ simplices.

${\cal T}_0 = {\cal P}$ is properly colored by $\chi: V({\cal P}) \rightarrow
\{0,\ldots,n-1\}$, i.e., each vertex is associated with the processor id which is its
color, and with its initial state.  Every refinement step consists of several $xy$-{\em
split} operations, that split each simplex of ${\cal T}_i$ into three simplices in ${\cal
T}_{i+1}$.  In an $xy$-{\em split}, where $x, y$ is a pair of adjacent vertices of ${\cal
T}$ we partition the edge $xy$ into three segments $x=z_0,z_1,z_2,z_3=y$ in this order.
All old vertices maintain their $\chi$ values, while $\chi(z_1):=\chi(y)$ and
$\chi(z_2):=\chi(x)$. Furthermore, every simplex $\sigma$ of ${\cal T}$ with $x,y \in
V(\sigma)$, say $\sigma={\rm conv}(\{x,y\}\dot\cup S)$ (i.e., $S \cup \{x,y\}$ is the set
of vertices of $\sigma$) is split accordingly to three simplices $\sigma_0 \cup \sigma_1
\cup \sigma_2$ where $\sigma_i = {\rm conv}(\{z_i,z_{i+1}\}\dot\cup S)$.  We note,
following Section \ref{sec:model} that an $xy$-{\em split} corresponds to a
rendezvous between processors $\chi(x)$ and $\chi(y)$. Vertex $x$ corresponds to the
state of processor $\chi(x)$ in ${\cal T}_i$  concatenated with it sending the message in
round $i+1$ rendezvous but not receiving any message, vertex $z_1$ corresponds to
processor $\chi(y)$ in ${\cal T}_i$ concatenated with it sending its state (message) in
round $i+1$ rendezvous and receiving the message from processor $\chi(x)$, and so forth.
To construct ${\cal T}_{i+1}$ from ${\cal T}_i$ we apply the $xy$-{\em split} operation
to all pairs of vertices, $x, y$ in ${\cal T}_i$, such that $(\chi(x), \chi(y))$ is the
unordered pair of processors associated with round $i+1$ in the TP-pairs schedule. See
Figure \ref{fig:2ad-1} for an example.

It remains to show that the coloring $\chi: V({\cal T}) \rightarrow \{0,\ldots,n-1\}$
satisfies Sperner's condition.  Clearly ${\cal P}$ satisfies the condition.   Consider by
contradiction the first time that the condition fails, in an $xy$-split operation $o$ in
which the edge $x_o$--$y_o$ is split, creating the two new vertices $z_{o1}$, $z_{o2}$,
with colors $\chi(x_o)$ and $\chi(y_o)$.  The only way the condition can fail is if
either of the $z$'s is not in the union of the carriers of $x_o$ and $y_o$.  But, by
simple algebra, the carrier of any point along the $x_o$--$y_o$ edge (line) belongs to
the union of the carriers of $x_o$ and $y_o$, thus $\chi(z_{o1})$ and $\chi(z_{o2})$ are
colors of vertices in their carrier, concluding the proof. \end{proof}


\begin{thebibliography}{99}

\bibitem{ATOMIC}
Afek Y., H. Attiya, Dolev D., Gafni E., Merrit M. and Shavit N.,
Atomic Snapshots of Shared Memory.
{\em Proc. 9th ACM Symposium on Principles of Distributed Computing (PODC'90)},
ACM Press, pp. 1--13, 1990.

\bibitem{AAFST}
Yehuda Afek, Hagit Attiya, Arie Fouren, Gideon Stupp and
Dan Touitou: Long-Lived Renaming Made Adaptive,
{\em PODC}, {1999}, {91-103}.

\bibitem{ae84}
Baruch Awerbuch, and Shimon Even,
Efficient and reliable broadcast is achievable in an eventually connected network(Extended Abstract),
{\em Proceedings of the third annual ACM symposium on Principles of distributed computing},
PODC '84, Vancouver, British Columbia, Canada,
278--281 (1984).

\bibitem{ag88}
Yehuda Afek and Eli Gafni,
End-to-end communication in unreliable networks,
{\em Proceedings of the seventh annual ACM Symposium on Principles of distributed computing},
PODC '88, Toronto, Ontario, Canada, 131--148, 1988.

\bibitem{collect1}
Yehuda Afek, Gideon Stupp and Dan Touitou:
{Long-lived Adaptive Collect with Applications},
{FOCS}, {1999}, {262-272}.

\ignore{%%%%%%%%%%%%%%%%

\bibitem{AWW93}
Yehuda Afek, Eytan Weisberger, Hanan Weisman,.
A Completeness Theorem for a Class of Synchronization Objects (Extended Abstract).
ACM PODC 1993: 159--170.
}%%%%%%%%%%%%%%%%%%%%%%%%%%%%

\bibitem{ABD}
Hagit Attiya, Amotz Bar-Noy, Danny Dolev,
Sharing Memory Robustly in Message-Passing Systems.
\emph{J. ACM} 42(1): 124--142 (1995).

\ignore{%%%%%%%%%%%%%%%%

\bibitem{RENAMING}
Hagit Attiya, Amotz Bar-Noy, Danny Dolev, David Peleg, RŸdiger Reischuk: Renaming in an
Asynchronous Environment, \emph{J. ACM} 37(3): 524--548 (1990).

\bibitem{BDDS92}
Amotz Bar-Noy, Danny Dolev, Cynthia Dwork, H. Raymond Strong: Shifting Gears: Changing
Algorithms on the Fly to Expedite Byzantine Agreement, \emph{Inf. Comput.} 97(2):
205--233 (1992).

}%%%%%%%%%%%%%%%%%%%%%%%%%%%%

\bibitem{BNaor}
Bar-Noy, A., Naor, J.: Sorting, Minimal Feedback Sets and Hamilton Paths in Tournaments,
{\em SIAM Journal on Discrete Mathematics} 3 (1): 7–20, 1990, doi:10.1137/0403002.

\bibitem{BG1}
Borowsky E. and Gafni E., Generalized FLP Impossibility Results for $t$-Resilient
Asynchronous Computations {\it Proc. 25th ACM Symposium on the Theory of Computing
(STOC'93)}, ACM Press, pp. 91-100, 1993.

\bibitem{BG93}
Borowsky E. and Gafni E., Immediate Atomic Snapshots and Fast Renaming (Extended
Abstract). PODC 1993: 41--51.

\bibitem{BG97}
Borowsky E. and Gafni E.,
A Simple Algorithmically Reasoned Characterization of Wait-Free
Computations (Extended Abstract).
{\it Proc. 16th ACM Symposium on Principles of Distributed Computing
(PODC'97)}, ACM Press, pp. 189--198, August 1997.


\bibitem{BG2}
Borowsky E., Gafni E., Lynch N. and Rajsbaum S.,
The {BG} Distributed Simulation Algorithm.
\emph{Distributed Computing,} 14(3):127--146, 2001.

\ignore{%%%%%%%%%%%%%%%%
\bibitem{CR08}
Armando Casta\~neda, Sergio Rajsbaum., New combinatorial topology upper and lower bounds
for renaming. PODC 2008: 295--304.

\bibitem{Chin-sung-gafni87}
Ching-Tsun Chou, Eli Gafni, Understanding and Verifying Distributed Algorithms Using
Stratified Decomposition. PODC 1988: 44--65.

\bibitem{DRandell}
John Dobson and Brian Randell,
Introduction to
``Building Reliable Secure Computing Systems out of Unreliable Insecure Components,''
Department of Computing Science,
University of Newcastle upon Tyne,
Newcastle NE1 7RU, U.K.

}%%%%%%%%%%%%%

\bibitem{DR85}
Danny Dolev, R\"{u}diger Reischuk: Bounds on information exchange for Byzantine agreement,
{\em JACM} 32, 1, January, 1985, 191--204.

\ignore{%%%%%%%%%%%%%%%%
\bibitem{CCL}
Tzilla Elrad, Nissim Francez,
Decomposition of Distributed Programs into Communication-Closed Layers.
{\em Sci. Comput. Program}. 2(3): 155--173 (1982). }%%%%%%%%%%%%%


\bibitem{FLP85}
Fischer M.J., Lynch N.A. and Paterson M.S.,
Impossibility of Distributed Consensus with One Faulty Process.
{\em Journal of the ACM}, 32(2):374--382, 1985.

\ignore{%%%%%%%%%%%%%%%%
\bibitem{F99}
Arie Fouren. Exponential examples for two renaming algorithms. Available at
\url{http://www.cs.technion.ac.il/~hagit/publications/expo.ps.gz}, Aug. 1999.}%%%%%%%%%%

\bibitem{G98}
Eli Gafni: Round-by-Round Fault Detectors,
Unifying Synchrony and Asynchrony (Extended Abstract). PODC 1998: 143--152.

\ignore{%%%%%%%%%%%%%%%%

\bibitem{01}
Eli Gafni,
The 0-1-Exclusion Families of Tasks.
OPODIS 2008: 246--258.

\bibitem{GMT01}
Gafni E., Merritt M. and Taubenfeld G.,
The Concurrency Hierarchy, and Algorithms for Unbounded Concurrency.
\emph{Proc. 21st ACM Symposium on Principles of Distributed Computing (PODC'01)},
ACM Press, pp. 161--169, 2001.

}%%%%%%%%%%%%%%%%%%%

\bibitem{SergioOpodis}
Eli Gafni, Sergio Rajsbaum, Distributed Programming with Tasks. \emph{OPODIS 2010}: 205-218

\ignore{%%%%%%%%%%%%%%%%

\bibitem{GRH06}
Eli Gafni, Sergio Rajsbaum, Maurice Herlihy,
Subconsensus Tasks: Renaming Is Weaker Than Set Agreement. DISC 2006: 329--338.

\bibitem{GHS}
Robert G. Gallager, Pierre A. Humblet, Philip M. Spira:
A Distributed Algorithm for Minimum-Weight Spanning Trees.
\emph{ACM Trans. Program. Lang. Syst.} 5(1): 66--77 (1983).

}%%%%%%%%%%%%%%%%%%%

\bibitem{H91}
Herlihy M.P.,
Wait-Free Synchronization.
{\it ACM Transactions on programming Languages and Systems}, 11(1):124--149, 1991.

\ignore{%%%%%%%%%%%%%%%%

\bibitem{HR2010}
Maurice Herlihy and Sergio Rajsbaum, The Topology of Shared-Memory Adversaries. 29th ACM
Symposium on Principles of Distributed Computing (PODC), Zurich, Switzerland, July
25--28, 2010. To appear. }%%%%%%%%%%%%%%%%%%%



\bibitem{hrt98}
Maurice Herlihy, Sergio Rajsbaum, Mark R. Tuttle, Unifying Synchronous and Asynchronous
Message-Passing Models. {\em PODC} 1998: 133-142.

\bibitem{hs99}
Herlihy M.P. and Shavit N.,
The Topological Structure of Asynchronous Computability.
{\em Journal of the ACM}, 46(6):858--923, 1999.

\bibitem{klo10}
Fabian Kuhn, Nancy Lynch, Rotem Oshman. Distributed Computation in Dynamic Graphs.
In the {\em 42nd ACM Symposium on Theory of Computing (STOC 2010).}

\bibitem{lamport86}
Leslie Lamport, On Interprocess Communication. Part II: Algorithms. {\em Distributed
Computing}, 1(2):86--101, 1986.

\bibitem{lan53}
H. Landau. On dominance relations and the structure of animal societies, III: The
condition for score structure. {\em Bulletin of Mathematical Biophysics}, 15(2):143–148,
1953.

\bibitem{moran}S. Moran, Y. Wolfstahl, Extended impossibility results for asynchronous complete
networks, {\em Inf. Process. Lett.} 26 (3)  145–151 1987.

\bibitem{mosesRajsbaum}
Yoram Moses, Sergio Rajsbaum: The Unified Structure of Consensus: A Layered Analysis
Approach. {\em PODC} 1998: 123-132

\ignore{%%%%%%%%%%%%%%%%

\bibitem{byzantine}
L. Lamport, R. Shostak, and M. Pease,
The Byzantine Generals Problem.
\emph{ACM Transactions on Programming Languages and Systems} 4 (3): 382--401, 1982.

\bibitem{manber}
Udi Manber,
\emph{Introduction to Algorithms: A Creative Approach,}
Addison Wesley, 1989.

\bibitem{randell}
Randell, Brian
Recursively structured distributed computing systems, IEEE
Symposium on Reliability in Distributed Software and Database Systems, 3--11, 1983.

\bibitem{RRT-ipl08}
Sergio Rajsbaum, Michel Raynal, Corentin Travers,
An impossibility about failure detectors in the iterated immediate snapshot model.
\emph{Inf. Process. Lett.} 108(3): 160--164 (2008).

\bibitem{RRT08}
Sergio Rajsbaum, Michel Raynal, Corentin Travers,
The Iterated Restricted Immediate Snapshot Model. COCOON 2008: 487--497.



\bibitem{redei}
La'szlo' Re'dei, {\em Acta Sci. Math.} (Szeged), 43 (1981) 1-2

}%%%%%%%%%%%%%%%

\bibitem{SZ00}
Saks, M. and Zaharoglou, F.,
Wait-Free $k$-Set Agreement is Impossible: The Topology of Public Knowledge.
{\em SIAM Journal on Computing}, 29(5): 1449--1483, 2000.

\ignore{%%%%%%%%%%%%%%%%%%%

\bibitem{S00}
Stojmenovic, I.,
Recursive algorithms in computer science courses: Fibonacci numbers
and binomial coefficients.
\emph{IEEE Trans. on Education,} 43(3): 273--276, Aug. 2000.

\bibitem{W94}
Hanan Weisman,
Implementing shared memory overwriting objects.
Master's thesis, Tel Aviv University, May 1994.

}%%%%%%%%%%%%%%%%%%

\end{thebibliography}
\end{document}